\documentclass{SCIS2023}

\begin{document}
\ArticleType{RESEARCH PAPER}
\Year{2022}
\Month{}
\Vol{}
\No{}
\DOI{}
\ArtNo{}
\ReceiveDate{}
\ReviseDate{}
\AcceptDate{}
\OnlineDate{}

\title{Constructions of Optimal Binary Locally Repairable Codes via Intersection Subspaces}{Constructions of {\color{blue}Optimal Binary} Locally Repairable Codes via Intersection Subspaces}



\author[1]{Wenqin Zhang}{}
\author[1]{Deng Tang}{dengtang@sjtu.edu.cn}
\author[1]{Chenhao Ying}{}
\author[1]{Yuan Luo}{luoyuan@cs.sjtu.edu.cn}

\AuthorMark{Wenqin Zhang}

\AuthorCitation{W.Zhang, D.Tang, C.Ying, et al }


\address[1]{ School of Electronic Information and Electrical Engineering, Shanghai Jiao Tong University, Shanghai, {\rm200240}, China.}


\abstract{Locally repairable codes (LRCs), which can recover any symbol of a codeword by reading only a small number of other symbols, have been widely used in real-world distributed storage systems, such as Microsoft Azure Storage and Ceph Storage Cluster. Since binary linear LRCs can significantly reduce coding and decoding complexity, constructions of binary LRCs are of particular interest.

The aim of this paper is to construct dimensional optimal binary locally repairable codes with disjoint local repair groups. We introduce how to connect intersection subspaces with binary locally repairable codes and construct dimensional optimal binary linear LRCs with locality $2^b$ ($b\geq 3$) and  minimum distance $d\geq 6$ by employing intersection subspaces deduced from the direct sum. This method will sufficiently increase the number of possible repair groups of dimensional optimal LRCs, and thus efficiently  expanding the range of the construction parameters while keeping the largest code rates compared with all known binary linear LRCs with minimum distance $d\geq 6$ and locality $2^b$.}


\maketitle

\section{Introduction}\label{introduction}
Efficient distributed storage systems (DSSs) provide access to data by storing it in a distributed manner across several storage nodes. Data loss and unavailability could happen in a DSS due to the unreliability of individual nodes. A classical technique used in the storage system is replication schemes. In such a scheme, copies of data packets are stored across different nodes. This scheme provides high reliability and availability. A disadvantage of this scheme is that replication has very high storage overhead. In the case of accelerated and relentless data growth, a new technique is necessary.

Erasure coding has been widely used in distributed storage systems, such as Windows Azure Storage \cite{Cheng2012} and Facebook Analytics Hadoop cluster \cite{Sathiamoorthy2013}, because of its higher fault-tolerance values and lower storage overheads. The failure node can be repaired by calculating the redundancy out of the original data over erasure channel. For an erasure code with length $n$, dimension $k$ and minimum distance $d$, any $d-1$ failures can be repaired by contacting at least $k$ other nodes.  Among these, the traditional maximum distance separable (MDS) erasure codes are optimal in terms of storage overhead. However, in the case of a single-node failure, the traditional MDS codes require connecting a large subset of surviving nodes, which will lead to an increase in the complexity of network traffic and the amount of input/output (I/O) operations. Consequently, regenerating codes \cite{Dimakis2010} and  codes with locality (known more commonly as locally repairable codes) \cite{Gopalan2012} were introduced in such a scenario. Regenerating codes can  efficiently repair a failure node by  minimizing the number of transmitted symbols. There are some works of regenerating codes in \cite{Rashmi2011,Hou2019,Zhou2022Explicit}.  Nevertheless, the number of nodes contacted for repair can be a bottleneck for the system efficiency. Hence, locally repairable codes (LRCs) were introduced to optimize the number of disk reads required to repair a single-node failure. This paper is devoted to the construction of locally repairable codes with disjoint repair groups and good parameters.
\subsection{Code with locality and known results}
Let $q$ be a power of an arbitrary prime and $\mathbb{F}_q$ be the finite field with $q$ elements.
Let $\mathcal{C}$ be an $[n, k, d]_q$ linear code with length $n$, dimension
$k$ and minimum distance $d$ over $\mathbb{F}_q$.
The code $\mathcal{C}$ is called an LRC with locality $r$ if each code symbol  $c_i$ in a codeword $\mathbf{c}\in \mathcal{C}$ can be recovered by downloading at most $r$ other symbols. We denote such a code $\mathcal{C}$ as an $r$-LRC. When $q=2$, we omit $q$ from the notation $[n, k, d]_q$. In addition, the set of such $r$ symbols that can repair the $i$-th symbol is called a ``repair set".

LRCs are well studied and many works have been done (see, for example, in
\cite{Matthias2017,Gopalan2012,Viveck2015,guruswami2019,Anyu_2019,Bin2021}) to explore the relationship between parameters $n,k,d,r$. For an $[n,k,d]_q$ LRC with locality $r$, Gopalan et al. \cite{Gopalan2012} proved the well-known Singleton-like bound as:
\begin{eqnarray}\label{eq:sin}
  d \leq n-k-\left\lceil \frac{k}{r}\right\rceil+2,
\end{eqnarray}
where $\lceil\cdot\rceil$ stands for the ceiling function. An LRC is said to be $d$-optimal if it satisfies~\eqref{eq:sin} with equality for given $n,k,r$.  When $r=k$, the bound~\eqref{eq:sin} specializes to the classical Singleton bound $d\leq n-k+1$. Over the past few years, many constructions of optimal LRCs which achieve bound~\eqref{eq:sin} have been presented. In \cite{Tamo2014}, Tamo et al. proposed optimal LRCs over a finite field of size $q\geq n+1$ via subcodes of Reed-Solomon codes.
In~\cite{Jie2016}, Hao et al. designed optimal LRCs with $d = 3, 4$ over a finite field of size $q\geq r+2$.  By automorphism groups of elliptic curves, \cite{Liming2020} constructed optimal locally repairable codes with length up to $q+\sqrt{q}$. In addition, the constructions of optimal LRCs based on Reed-Solomon codes, among other techniques, have been recently discovered in \cite{tamo2016,Lingfei2019,Zhang2020}.
Notice that bound~\eqref{eq:sin} has been proved to be tight for some special cases
with large alphabet size according to the construction provided in~\cite{Gopalan2012}.
When it comes to small fields, parameters of the optimal constructions become very restrictive~\cite{Zhang2017},~\cite{hao2017optimal}.

In practice, codes over small alphabets attract more attention particularly in the application of storage because of their ease for implementation. In 2013, Cadambe and Mazumdar derived a new bound for $[n,k,d]_q$ LRCs which took the size of the alphabet into account~\cite{Viveck_2013}. This bound is known as C-M bound. They showed that the dimension $k$ of an $[n, k, d]_q$ LRC with locality $r$ is upper bounded by
\begin{equation}\label{eq:2}
  k \leq \min _{t \in \mathbb{Z}^{+}}\left\{t r+k_{opt}^{(q)}(n-t(r+1), d)\right\},
\end{equation}
where $k_{opt}^{(q)}(n, d)$ is the largest possible dimension of a code with length $n$ for a given alphabet size $q$ and a given minimum distance $d$. This bound applies to both linear and nonlinear codes.  Later in~\cite{Viveck_2013,Natalia2015}, explicit constructions of the family of binary LRCs are proposed which achieve the bound~\eqref{eq:2}. However, because the exact value of $k^{(q)}_{opt}(n,d)$ can only be obtained in a limited case with relatively short code length, it is difficult to apply the C-M bound to evaluate the optimality of general LRCs.

The original LRCs only support repair of a single-node failure. To address the problem of multiple nodes failures in practical scenarios, the concept of original LRCs was further generalized to LRCs with $(r,\delta)$-locality by Prakash et al. \cite{prakash2012optimal}. When $\delta=2$, an LRC with $(r,2)$-locality is reduced to an LRC with locality $r$. In 2019,  Grezet M. et al. \cite{grezet2019alphabet} used consecutive residual codes and Griesmer bound $\mathcal{G}(k,d)=\sum_{i=0}^{k-1}\left\lceil\frac{d}{q^{i}}\right\rceil\leq n$
 to derive a new alphabet-dependent bound for an $(r,\delta)$-LRC with parameters $[n,k,d]$:
 \begin{equation}\label{eq:gre}
     k\leq \min_{\ell\in \mathbb{Z}^{+}}\{\ell+k^{(q)}_{opt}(n-(a+1)\mathcal{G}(\kappa,\delta)+\mathcal{G}(\kappa-b,\delta),d)\},
 \end{equation}
 where $\kappa$ is the upper bound on dimension of local codes and $a, b\in \mathbb{Z}$ satisfy $\ell=a\kappa+b,0\leq b\leq \kappa$. Notably, the bound~\eqref{eq:gre} is tighter than the bound C-M bound when $\kappa<r$ or $\delta>q$. Consequently, one tends to use the bound~\eqref{eq:gre} for binary $(r, \delta)$-LRCs with $\delta>2$. In the past decade, many results have been obtained for $(r,\delta)$-LRCs \cite{Song2014,Han2020,xing2022,Luo2022,cai2022optimal}. In addition, a lot of progress on the
study of derivatives of  LRC has been made, such as LRCs with locality and availability \cite{wang2014,Tan2020,Jin2022},  Maximally  recoverable LRCs~\cite{Cheng2012,Dhar2023}, Scalable local reconstruction code~\cite{Zhang2022scala}.


Recently, Wang et al.~\cite{Anyu_2019} presented a sphere-packing
bound for binary LRCs based on disjoint local repair groups, which serves as a generalization
of the bounds in \cite{ Gopalan2012,Alexander2015}. Consider binary linear LRCs with minimum distance $d\geq5$: the dimension $k$ is actually upper bounded by the largest integer no greater than the following  explicit bound~\cite{Anyu_2019} given in equation~\eqref{eq:3}. For any $[n, k, d]$ binary linear LRCs with locality $r$ such that $d \geq 5$ and $2 \leq r \leq \frac{n}{2}-2$, it follows that
\begin{equation}\label{eq:3}
k \leq \frac{r n}{r+1}-\min \left\{\log _{2}\left(1+\frac{r n}{2}\right), \frac{r n}{(r+1)(r+2)}\right\}.
\end{equation}
We say that a binary linear LRC is $k$-optimal if it satisfies the bound~\eqref{eq:3} with equality for given $n$, $d$ and $r$. This paper will focus on a general assumption $n \geq 5(r+1)(r+2)$ that will be satisfied in the main results, and thus the bound~\eqref{eq:3} can be further simplified to be
\begin{equation*}
k \leq \frac{r n}{r+1}-\log _{2}\left(1+\frac{r n}{2}\right) .
\end{equation*}


 Compared to $q$-ary LRCs, binary LRCs are known to be advantageous in terms of implementation complexity in practical systems. For optimal binary linear LRCs, various construction methods  have been proposed, especially for the case of minimum distance $3$ and $4$.
Nevertheless, constructing optimal binary LRCs becomes increasingly challenging as the minimum distance requirement grows. In 2017, by the partial spread, Nam et al. constructed a class of binary linear LRCs with minimum distance at least $6$ and showed some examples that are optimal with respect to the bound~\eqref{eq:2}~\cite{Nam2017}. Subsequently,
 Wang et al. in~\cite{Anyu_2019} constructed an $\left[n=\frac{2^{s}-1}{2^{t}-1}, k=\frac{r n}{r+1}-s, d \geq 6\right]$ binary $k$-optimal LRC with locality $r = 2^b$ from generalized Hamming codes, where $s$ and $t$ are integers that satisfy $2t|s$ and $\frac{s}{2t}\geq 2$. In~\cite{Jingxue2019}, Ma et al. proposed a class of $k$-optimal binary linear LRCs for $d = 6$, which included the codes given in~\cite{Anyu_2019}, and they also presented a new $k$-optimal construction for locality $3$ and minimum distance $6$ from a partial $t$-spread.  For any fixed locality $r$ and minimum distance $d$, the coding rate of optimal LRCs becomes larger as the code length becomes larger~\cite{Dimakis2010}. Note that most constructions of binary linear LRCs are based on the partial $t$-spread, namely a set of mutually disjoint $t$-dimensional subspaces. Owing to the mutually  disjoint subspaces, it is easy to calculate the minimum distance of binary linear LRCs. However, a disadvantage of this approach is that the code length is limited. Actually, for a given locality, based on the intersection subspace, $k$-optimal binary linear LRCs can be constructed with code length larger than previously known, but few constructions exist. A more ingenious approach is necessary to cope with the intersection subspace in order to guarantee the minimum distance. Hence, it is a challenging and interesting problem to construct $k$-optimal binary linear LRCs by applying intersection subspaces.

 \subsection{Our results}
This paper focuses on a single-node failure problem of LRCs. We consider binary linear LRCs of which all local repair groups have uniform size $r + 1$ and are pairwise disjoint, i.e., $(r + 1)|n$.
 Using parity check matrices, we present an explicit construction of binary linear LRCs based on intersection subspaces with minimum distance $d\geq6$ and locality $r=2^b$. These intersection subspaces are designed by the direct sum of subspaces. Our LRCs turned out to be $k$-optimal in terms of the bound~\eqref{eq:3}. Precisely speaking, the following results are obtained.

 An explicit construction of $k$-optimal binary linear LRCs with new parameters $[n=(r+1)\ell, k\geq n-s-\ell-m, d\geq6]$ is proposed (see Construction~\ref{cons-lrc} and Theorem~\ref{Thm:paras} below), where $\ell= \frac{2^m-1}{2^{2b-s}-1}$, $b\geq3$ , $0\leq s< b$ and $(2b-s)|m$. When $\ell$ belongs to a determined range, those binary linear LRCs all can attain the bound~\eqref{eq:3}, so they are $k$-optimal (Theorem~\ref{tm:op-l}). In the case of $(2b-s)\nmid m$, we construct $k$-optimal binary linear LRCs with parameters $[n=(r+1)\ell, k=n-\ell-s-m,d\geq6]$  where $\ell=\frac{2^{m-s}-2^{(2b-s)}(2^z-1)-1}{2^{(2b-s)}-1}$, $0\leq s< b$ and $z\equiv (m-s)\bmod (2b-s)< b$ (Theorem~\ref{th3:parity}). Similar to Theorem~\ref{tm:op-l}, a class of $k$-optimal LRCs with a wider code length can be obtained from Theorem~\ref{th3:parity} (see details in Theorem~\ref{th4:l}). All results of $k$-optimal binary linear LRCs in this paper are summarized in Table~\ref{tab:con}.
 \begin{table*}[ht]
  \small{\centering
  \caption{ $k$-optimal binary linear LRCs with $d\geq6$ and $r=2^b$}\label{tab:con}
 \begin{tabular}{|p{2cm}|p{3cm}|p{4.3cm}|p{4.4cm}|}
   \hline
    & $ n,k$ & $\ell$ & $b\geq3,m\geq4b,s$ \\
  \hline \multirow{2}{*}{Theorem~\ref{Thm:paras}}& $n=(r+1)\ell$, & \multirow{2}{*}{$\ell=\frac{2^m-1}{2^{2b-s}-1}$} & $0\leq s< b$, $(2b-s)|m$   \\
  &$k= n-s-\ell-m$&&\\
  \hline \multirow{2}{*}{Theorem~\ref{tm:op-l}} & $n=(r+1)\ell$,  &
    \multirow{2}{*}{$\frac{2^{m+s-1}-1}{2^{b-1}(2^{b}+1)}<\ell \leq \frac{2^m-1}{2^{2b-s}-1}$} & $0\leq s< b$, $(2b-s)|m$  \\
    &$k= n-s-\ell-m$&&\\
    \hline \multirow{2}{*}{Theorem~\ref{th3:parity}} &  $n=(r+1)\ell$, &  \multirow{2}{*}{$\ell=\frac{2^{m}-2^{(2b-s)}(2^z-1)-1}{2^{(2b-s)}-1}$}& $0\leq s< b$, $(2b-s)\nmid m$,\\
   &$k= n-\ell-m-s$&&$z\equiv m\bmod (2b-s)\leq b$ \\
       \hline \multirow{3}{*}{Theorem~\ref{th4:l}} &  $n=(r+1)\ell$, &  $\frac{2^{m+s-1}-1}{2^{b-1}(2^{b}+1)}<\ell\leq$& $0\leq s< b$, $(2b-s)\nmid m$,\\
   &$k= n-\ell-m-s$&$\frac{2^{m}-2^{(2b-s)}(2^z-1)-1}{2^{(2b-s)}-1}$&$z\equiv m\bmod (2b-s)\leq b$ \\ \hline
 \end{tabular}}
 \end{table*}
 Moreover, we compare our results with the state-of-the-art approaches for a fixed locality $r$. 
 The results show that the $k$-optimal LRCs in this work have more flexible parameters  $[n,k]$ than those in \cite{Anyu_2019,Jingxue2019}. In other words, our construction can generate more repair groups, so with the same locality, the code length of $k$-optimal LRCs is larger. Additionally, by calculating code rates, it can be obtained that the code rate $R \triangleq k/n$ in our construction is higher than that in \cite{Anyu_2019}. At the end of this paper, a shortening technique will yield the derivation of new binary linear LRCs. By deleting codewords in $k$-optimal binary linear LRCs with nonzero values in the last coordinates and then removing the last coordinates from the remaining
codewords, we can suggest new parameters from the original binary linear LRCs (Theorem~\ref{the:shorten}).
\subsection{Organization}

In Section \ref{preliminarie}, some basic definitions and
results on LRCs, partial $t$-spread and intersection subspace are introduced. In Section \ref{se3}, we present a definition of a desired matrix and an explicit construction of LRCs. Based on this construction, we obtain the main results in Theorem~\ref{Thm:paras}, Theorem~\ref{tm:op-l}, Theorem~\ref{th3:parity} and Theorem~\ref{th4:l}. We also give three examples to explain the corresponding construction and some tables to show the comparison. In Section \ref{sec:5}, Theorem~\ref{the:shorten} proposes the result to shorten binary linear LRCs. Finally, Section \ref{sec:6} concludes the paper. In addition, the constructions of some desired matrices are displayed in~\ref{sec:8}.

\section{Preliminaries}\label{preliminarie}
In this section, we introduce some notations and basic results required later in this paper.

\begin{itemize}
\item Let $V_{n}(q)$ be the vector space with dimension $n$ over $\mathbb{F}_q$. When $q = 2$, we omit $q$ from the notation $V_{n}(q)$.
\item Suppose that $n$ is a  positive integer, we write $[n]=\{1,\cdots,n\}$.
\item For any $\mathbf x=(x_1, x_2,\cdots , x_n)$ and $\mathbf y=(y_1, y_2,\cdots , y_n)$ in $V_n(q)$,  the
Euclidean inner product of $\mathbf x$ and $\mathbf y$ is defined as $\mathbf x \cdot \mathbf y=\sum_{i=1}^n x_i y_i$.
\item  The support set of $\mathbf{x}$ is denoted by $\operatorname{supp}(\mathbf{x}) = \{i~|~x_i\neq 0\}$.

\end{itemize}

\subsection{Locally repairable codes}
Let $\mathcal{C}$ be an $[n,k,d]_q$ linear code. Then, $\mathcal{C}$ has a $k \times n$ generator matrix $G$ and an $(n-k) \times n$ parity check matrix $H$. The dual of $\mathcal{C}$ is defined by
\begin{equation*}
\mathcal{C}^{\perp}=\left\{\mathbf{w} \in V_n(q): \mathbf{w} \cdot \mathbf{c}=0 \text { for all } \mathbf{c} \in \mathcal{C}\right\}.
\end{equation*}
The rows of $H$ are codewords of $\mathcal{C}^{\perp}$. Hence, the $k\times n$ generator matrix $G$ and $(n-k)\times n$ parity check matrix $H$ satisfy $GH^T=\mathbf{0}$, where $T$ denotes the transpose of a matrix. There is a well-known distance property of linear codes
as follows.

\begin{lemma}[\cite{ling2004},Theorem 4.5.6 ]
Let $\mathcal{C}$ be a linear code and let $H$ be a parity check matrix for $\mathcal{C}$. Then the minimum distance of $\mathcal{C}$ is not less than $d$ if and only if any $d-1$ columns of $H$ are linearly independent.
\end{lemma}
 Now, we give the formal definition of linear LRCs.

\begin{definition}\label{def:LRC}
 The linear code $\mathcal{C}$ is a locally repairable code (LRC) with locality $r$ if for any $i\in[n]$, there exists a subset $\mathcal R_i\subset[n]\backslash\{i\}$ with $|\mathcal R_i|\leq r$ such that the $i$-th symbol $c_i$ in each codeword $\mathbf{c}=(c_1,c_2,\cdots,c_n)\in \mathcal{C}$ can be recovered by $\{c_j\}_{j\in \mathcal R_i}$, i.e.,
$c_i$ is a linear combination of $\{c_j\}_{j\in \mathcal R_i}$.
The set $\mathcal R_i$ is called a repair set for $c_i$.

\end{definition}

It is well known that two different approaches are used to construct LRC, the generator matrix approach \cite{Gopalan2012} and the parity check matrix approach \cite{Jie2016}. Next, we will introduce the parity check matrix approach to construct LRCs. In order to find a suitable parity check matrix involving locality, we begin with a simple lemma.

\begin{lemma}[\cite{guruswami2019}]\label{lem:1}
 An LRC has locality $r$ if and only if for every coded symbol there exists a codeword $\mathbf{x}$ in $\mathcal{C}^\perp$ whose support set $\operatorname{supp}(\mathbf{x})$ contains $i$ and the size of $\operatorname{supp}(\mathbf{x})$ is at most $r + 1$.
\end{lemma}

 An LRC is said to have $\ell$  disjoint local repair groups if there exist $\ell \triangleq \frac{n}{r+1}$ vectors $\mathbf{h}_1, \mathbf{h}_2, \cdots, \mathbf{h}_\ell$ of $\mathcal{C}^{\perp}$, such that $|\operatorname{supp} (\mathbf{h}_{i})|=r+1$ and $\operatorname{supp}(\mathbf{h}_{i}) \cap \operatorname{supp}(\mathbf{h}_{j})=\emptyset$ for any
$1 \leq i \neq j \leq \ell$. Let $\mathcal{C}$  be an $[n,k]$ binary linear LRC with disjoint local repair groups. The parity check matrix $H$ of $\mathcal{C}$ can be represented as follows:
\begin{equation}\label{eq:pc}
H=\left(\begin{array}{llll}
H_{L} \\
H_G
\end{array}\right)=\left(\begin{array}{llll}
H_{1} & H_{2} & \ldots &H_{\ell} \\
H^{1}_G & H^{2}_G & \ldots & H^{\ell}_G
\end{array}\right),
\end{equation}
where $H_{i}$ is an $\ell \times(r+1)$ submatrix of whose $i$-th row is the all-one vector and the other rows are all-zero vectors and  $H^i_G$ is an $i$-th $(n-k-\ell)\times(r+1)$ submatrix of $H_G$ for $1 \leq i \leq \ell$. Note that, in the following sections, the column vector $\mathbf{h}^i_j$ denotes a column of $H^i_{G}$, which is different from the meaning of $\mathbf{h}_j$.

\subsection{Intersection subspace}
The set of all $t$-dimensional subspaces  of $V_m(q)$ is denoted by $\mathcal{G}_q(m, t)$. Let $U$ and $V$ be subspaces of $V_m(q)$. Then $U\bigcap V=\{\mathbf{v}|\mathbf{v}\in U~and~\mathbf{v}\in V\}$ is called the intersection of $U$ and $V$. It is clear that the intersection  $U\bigcap V$ is also a subspace.

 In particular,
two $t$-dimensional subspaces $U$ and $V$ which belong to  $\mathcal{G}_q(m, t)$ are said to trivially intersect or disjoint if they only have a zero-dimensional intersection, i.e.,  $U\bigcap V=\{\mathbf{0}\}$. A partial $t$-spread of $V_m(q)$  is a collection $S=\{W_1,W_2,\cdots,W_\ell\}$ of $t$-dimensional subspaces from $\mathcal{G}_q(m, t)$ such that $W_i\bigcap W_j=\{\mathbf{0}\}$ for $1\leq i<j\leq \ell$, where $\ell$ is the size of the partial $t$-spread $S$. If $t$ divides $m$ and $span(\bigcup^{\ell}_{i=1}W_i)= V_m(q)$, the partial spread is called a $t$-spread. We denote the number of $t$--dimensional subspaces in the largest partial spread in $V_m(q)$ by $\mu_q(m,t)$. One challenging question is to find the maximum partial size of a $t$-spread. There are few results related to $\mu_q(m,t)$, see below.

\begin{lemma}[\cite{Tor1989}]\label{lam:1}
If $t$ is a divisor of $m$ and $\ell=\frac{q^m-1}{q^t-1}$, then there exists a $t$-spread of $V_m(q)$ with $\ell$ subspaces.
\end{lemma}
\begin{lemma}[\cite{Etzion2011}]\label{lam:2}
 Let $m\equiv z\bmod t$. Then, for all $q$, we have
\begin{equation}\label{eq:4}
 \mu_q(m,t)\geq\frac{q^m-q^t(q^z-1)-1}{q^t-1}.
\end{equation}
\end{lemma}
Note that the Lemma~\ref{lam:1} is a special case of the Lemma~\ref{lam:2} if $z=0$. In addition, a specific construction for a $t$-spread is given in  \cite{Tor1989} and a specific construction for a partial $t$-spread is given in \cite{Etzion2011}.

\section{Construction by the intersection subspace}\label{se3}
In this section, we begin with a lemma that is essential to construct the parity check matrix of binary linear LRCs with minimum distance $d\geq6$. Then a definition of a desired matrix is given for subsequent constructions of LRCs. Finally, combining the intersection subspaces with a desired matrix generates $k$-optimal binary linear LRCs with disjoint local repair groups. The parameters of $k$-optimal LRCs are derived in Theorem~\ref{Thm:paras} and Theorem~\ref{th3:parity}, respectively.


\begin{lemma}[ \cite{Nam2017}]\label{lem:condition} Consider a binary linear LRC defined by the parity check matrix $H$ in~\eqref{eq:pc}. If the columns of $H_G$ satisfy the following three conditions, then the LRC has minimum distance $d\geq6$ :
\begin{itemize}
\item [(1)]No two column vectors from matrix $H_{G}^{i}$  sum to zero for all $i\in[\ell]$;
\item [(2)]No four column vectors from matrix $H_{G}^{i}$  sum to zero for all $i\in[\ell]$;
\item [(3)]No four column vectors consisting of two columns from matrix $H_{G}^{i}$ and the other two columns from matrix $H_{G}^{j}$ sum to zero for all distinct $i\neq j\in[\ell]$.
\end{itemize}
\end{lemma}

Hence, to obtain a binary linear LRC with minimum distance $d\geq6$, we need to construct the parity check matrix $H$ in~\eqref{eq:pc} of which submatrix $H^i_G$ satisfies the conditions of Lemma~\ref{lem:condition}.
%

 Suppose $W_{1}$, $W_{2}$, $\cdots$, $W_{\ell}$ are  $t$-dimensional subspaces of a vector space $V_m$  such that
$W_i\bigcap W_j= \{\mathbf{0}\}$  for $i\neq j \in[\ell]$. Let $\{\mathbf{e}^i_1,\mathbf{e}^i_2,\cdots,\mathbf{e}^i_t\}$  be a basis of the subspace $W_i$, where $\mathbf{e}^i_j=(e^i_{1j},e^i_{2j},\cdots,e^i_{mj})^{T} \in V_m$ for $j\in[t]$. We shall write the coordinates of the vector $\mathbf{e}^i_{j}$ as the $j$-th column of an $m\times t$  matrix $G_{W_i}$, i.e., $G_{W_i}=[\mathbf{e}^i_1,\mathbf{e}^i_2,\cdots,\mathbf{e}^i_t]$.

\begin{definition}\label{inter_good}
Let $G_U=[\mathbf{u}_1,\mathbf{u}_2,\cdots,\mathbf{u}_s]$ be an $s\times s$  matrix over $\mathbb{F}_2$ with full column rank and column vectors $\mathbf{u}_i=(u_{1i},u_{2i},\cdots,u_{si})^T\in V_s$ for $i\in[s]$. Let $U$ be the span of the columns of $\small{\left(\begin{array}{c}
G_U \\
\mathbf{0}_{m\times s}
\end{array}\right)}$, where  $\mathbf{0}_{m\times s}\in V_m$ is an $m\times s$ all-zero matrix. Define an $(s+m)\times(s+t)$ matrix $G_{M_i}$ over $\mathbb{F}_2$ as follows:
\begin{equation}\label{eq:G_M}
G_{M_i}=\left(\begin{array}{cccc}
G_U&  \mathbf{0}_{s\times t}\\
\mathbf{0}_{m\times s}& G_{W_i}
\end{array}\right)=\left(\begin{array}{ccccccccc}
\mathbf{u}_1 & \mathbf{u}_2 & \cdots & \mathbf{u}_s & \mathbf{0}_{s\times 1}
 & \mathbf{0}_{s\times 1} & \cdots & \mathbf{0}_{s\times 1} \\
 \mathbf{0}_{m\times 1} & \mathbf{0}_{m\times 1} & \cdots & \mathbf{0}_{m\times 1} & \mathbf{e}^i_{1} & \mathbf{e}^i_{2} & \cdots & \mathbf{e}^i_{t} \\
\end{array}\right).
\end{equation}
In the special case that $s=0$, we have $G_{M_i}=G_{W_i}$. Furthermore, the vector space spanned by the columns of $G_{M_i}$ is denoted as $M_i$.
\end{definition}

\begin{remark}
In this paper, we utilize the intersection subspaces $M_i$ for $i\in[\ell]$ to construct the parity check matrix of a binary linear LRC. Furthermore, combining a $t$-spread with a matrix $G_U$, we give an explicit form of such intersection subspaces. From the above construction, each subspace $M_i$ can be generated by a direct sum of two subspaces $U$ and $W_i$ for $i\in[\ell]$, i.e., $M_i=W_{i} \oplus U$. It is not hard to see that $M_i$ is a vector subspace of $V_{m+s}$ with dimension $s+t$. Hence, it is clear that the intersection of $M_1$, $M_2$, $\cdots$, $M_\ell$ only consists of the vector subspace $U$ with dimension $s$.
  \end{remark}

 Later, we will exploit these intersection subspaces and a matrix with a special structure to construct binary linear LRCs. In this paper, we call such a matrix a desired matrix
and give the definition of a desired matrix as follows.

 \begin{definition}\label{def:good-matrix}
Let $r,s$ and $t$ be integers such that $2^t>r> s+t$ and $t>s\geq0$. A binary $(s+t)\times r$ matrix $A=\left(\begin{array}{cccc}A_{1}\\ A_{2} \end{array}\right)=(a_{i,j})\in V_{(s+t)\times r}$ is a  \textbf{desired matrix} if it is a binary matrix with full column rank such that the submatrix $A_2$ has $t$ rows and no two linearly dependent columns; furthermore, if $t\geq3$, then any $4$ columns of $A$ must be linearly independent. Note that $A$  can be viewed as the parity check matrix of an $[r, r-(s+t), \geq5]$ binary linear code for $t\geq3$ (See more details in~\ref{sec:8}).

\end{definition}
\begin{remark}
The key idea of intersection subspaces comes from the construction of matrix $G_{M_i}$,
while matrix $A$ will help to expand the range of the construction parameters and control the minimum distance of binary linear LRCs.
\end{remark}

 Let $G_{M_i}$ be a matrix as defined in equation~\eqref{eq:G_M} for $i\in[\ell]$. Let $A$ be an $(s+t)\times r$ desired matrix. Then the product of each matrix $G_{M_i}$ and $A$ is
{\small\begin{eqnarray}
  G_{M_i}\cdot A
   &=&  \left(\begin{array}{ccccccccc}
u_{11}& u_{12} &\cdots & u_{1s} & \\
u_{21}& u_{22} &\cdots & u_{2s} &\\
\vdots & \vdots & & \vdots & &\multicolumn{2}{c}{\raisebox{1.3ex}[0pt]{\large $\mathbf{0}_{s\times t}$}}\\
u_{s1}& u_{s2} &\cdots & u_{ss} &   \\
&  & && e^i_{11} &e^i_{12} &\cdots & e^i_{1t}\\
&  & && e^i_{21} &e^i_{22} &\cdots & e^i_{2t}\\
&  \multicolumn{2}{c}{\raisebox{1.3ex}[0pt]{ \large $\mathbf{0}_{m\times s}$}}  &  & \vdots & \vdots & & \vdots \\
&  & && e^i_{m1} &e^i_{m2} &\cdots & e^i_{mt}\\
\end{array}\right)\left(\begin{array}{cccc}
a_{11}&a_{12}&\cdots&a_{1r}\\ \nonumber
a_{21}&a_{22}&\cdots&a_{2r} \\
\vdots&\vdots&&\vdots \\
a_{(s+t)1}&a_{(s+t)2}&\cdots&a_{(s+t)r}
\end{array}\right),
\end{eqnarray}}
\noindent where $\mathbf{0}$ denotes a zero matrix.
Let $H^i_G$ denote an $(s+m)\times(r+1)$ matrix
over $\mathbb{F}_2$ which is equal to the matrix $\left(\mathbf{0}_{(s + m)\times1},~G_{M_i}\cdot A\right)$.
Note that every column $\mathbf{h}^i_{j}$ of $H^i_G$ can be indexed by a pair $(i, j)$ for $1 \leq i \leq \ell$ and $0 \leq j \leq r$, where $\mathbf{h}^i_{0}=(0,0,\cdots,0)^T$ and
\begin{equation*}
\mathbf{h}^i_{j\neq0}=(\sum^s_{\ell=1}u_{1\ell}a_{\ell j},\cdots,\sum^s_{\ell=1}u_{s\ell}a_{\ell j}, \sum^t_{\ell=1}e^i_{1\ell}a_{(s+\ell)j},\cdots, \sum^t_{\ell=1}e^i_{m\ell}a_{(s+\ell)j})^{T}.
\end{equation*}
 Then, we propose the following Lemma~\ref{lem:independet} with regard to $H^i_G$, which plays an important role for the main results of Theorem~\ref{Thm:paras}.

\begin{lemma}\label{lem:independet}Let $H^{i}_G=\left(\mathbf{0}_{(s + m)\times1},~G_{M_i}\cdot A\right)$ for $i\in [\ell]$.
The sum of any four columns of $H_G$, which consists of two columns from matrix $H^{i_1}_G$ and the other two columns from matrix $H^{i_2}_G$, is not equal to zero for all distinct $i_1, i_2 \in[\ell]$.
\end{lemma}

\begin{proof}
 Considering the four columns $\mathbf{h}^{i_1}_{j_1}, \mathbf{h}^{i_1}_{j_2}\in H^{i_1}_G$, $\mathbf{h}^{i_2}_{j_3}, \mathbf{h}^{i_2}_{j_4}\in H^{i_2}_G$, without loss of generality, we assume that
$\mathbf{h}^{i_1}_{j_1}+\mathbf{h}^{i
_1}_{j_2}+\mathbf{h}^{i_2}_{j_3}+\mathbf{h}^{i_2}_{j_4}=\mathbf{0}$.

Case (i): The four columns $\{\mathbf h^{i_1}_{j_1},\mathbf h^{i_1}_{j_2},\mathbf h^{i_2}_{j_3},\mathbf h^{i_2}_{j_4}\}$ contain the zero column vector.

1)~Assume that one of the columns is the zero column from $H^{i_1}_G$, i.e., $\mathbf{h}^{i_1}_{j_1}=\mathbf{0}$. Then we have
\begin{equation}\label{eq:07091}
\left\{
  \begin{array}{ll}
   \sum^{s}_{\ell=1}u_{1\ell}a_{\ell j_2}+\sum^{s}_{\ell=1}u_{1\ell}a_{\ell j_3}+\sum^{s}_{\ell=1}u_{1\ell}a_{\ell j_4}=0 \\
   ~~~~~~~~~~~~~~~~~~~\vdots \\
\sum^{s}_{\ell=1}u_{s\ell}a_{\ell j_2}+\sum^{s}_{\ell=1}u_{s\ell}a_{\ell j_3}+\sum^{s}_{\ell=1}u_{s\ell}a_{\ell j_4}=0  \\
\sum^{t}_{\ell=1}e^{i_1}_{1\ell}a_{(s+\ell)j_2}+\sum^{t}_{\ell=1}e^{i_2}_{1\ell}a_{(s+\ell)j_3}+\sum^{t}_{\ell=1}e^{i_2}_{1\ell}a_{(s+\ell)j_4}=0\\
   ~~~~~~~~~~~~~~~~~~~\vdots \\
 \sum^{t}_{\ell=1}e^{i_1}_{m\ell}a_{(s+\ell)j_2}+\sum^{t}_{\ell=1}e^{i_2}_{m\ell}a_{(s+\ell)j_3}+\sum^{t}_{\ell=1}e^{i_2}_{m\ell}a_{(s+\ell)j_4}=0\\
  \end{array}.
\right.
\end{equation}
From~\eqref{eq:07091}, the last $t$ equations out of $m$ equations show that
\begin{equation*}\label{111}
\sum^t_{\ell=1}a_{(s+\ell) j_2}\left(\begin{array}{cccc}e^{i_1}_{1\ell}\\ e^{i_1}_{2\ell}\\ \vdots \\ e^{i_1}_{m\ell} \end{array}\right)+
\sum^t_{\ell=1}(a_{(s+\ell) j_3}+a_{(s+\ell) j_4})\left(\begin{array}{cccc}e^{i_2}_{1\ell}\\ e^{i_2}_{2\ell}\\ \vdots \\ e^{i_2}_{m\ell} \end{array}\right)=\mathbf{0}.
\end{equation*}
Since $W_{i_1}$ and $W_{i_2}$ are disjoint $t$-dimensional subspaces, i.e., $W_{i_1}\bigcap W_{i_2}=\{\mathbf{0}\}$, which means that the linear combinations of the basis $\{\mathbf {e}^{i_1}_1,\mathbf {e}^{i_1}_2,\cdots, \mathbf {e}^{i_1}_t\}$ of $W_{i_1}$ and $\{\mathbf {e}^{i_2}_1,\mathbf {e}^{i_2}_2,\cdots, \mathbf {e}^{i_2}_t\}$ of $W_{i_2}$ are linearly independent. Thus, $\sum^t_{\ell=1}a_{(s+\ell) j_2}\mathbf{e}^{i_1}_{\ell}=\mathbf{0}$ and $\sum^t_{\ell=1}(a_{(s+\ell) j_3}+a_{(s+\ell) j_4})\mathbf{e}^{i_2}_{\ell}=\mathbf{0}$, implying that $a_{(s+\ell) j_2}=0$ and $a_{(s+\ell) j_3}+a_{(s+\ell) j_4}=0$ for all $\ell\in[t]$. A contradiction can be obtained from the definition of the desired matrix $A$, so
 $\mathbf{h}^{i_1}_{j_1}+\mathbf{h}^{i
_1}_{j_2}+\mathbf{h}^{i_2}_{j_3}+\mathbf{h}^{i_2}_{j_4}\neq\mathbf{0}$.

2) In the case that $\mathbf h^{i_1}_{j_1}$ from $H^{i_1}_G$ and  $\mathbf{h}^{i_2}_{j_3}$ from $H^{i_2}_G$ is zero vector respectively, by the similar analysis of case 1), the same result holds.

Case (ii): The four columns $\{\mathbf h^{i_1}_{j_1},\mathbf h^{i_1}_{j_2},\mathbf h^{i_2}_{j_3},\mathbf h^{i_2}_{j_4}\}$ do not contain the zero column vector. Similar to the case (i), we have
\begin{equation}\label{eq:s}
\sum^s_{\ell=1}(a_{\ell j_1}+a_{\ell j_2}+a_{\ell j_3}+a_{\ell j_4})\left(\begin{array}{cccc}u_{1\ell}\\ u_{2\ell}\\ \vdots \\ u_{s\ell} \end{array}\right)=\mathbf{0},
\end{equation}
and
\begin{equation*}\label{eq:e}
\sum^t_{\ell=1}(a_{(s+\ell) j_1}+a_{(s+\ell) j_2})\left(\begin{array}{cccc}e^{i_1}_{1\ell}\\ e^{i_1}_{2\ell}\\ \vdots \\ e^{i_1}_{m\ell} \end{array}\right)+
\sum^t_{\ell=1}(a_{(s+\ell) j_3}+a_{(s+\ell) j_4})\left(\begin{array}{cccc}e^{i_2}_{1\ell}\\ e^{i_2}_{2\ell}\\ \vdots \\ e^{i_2}_{m\ell} \end{array}\right)=\mathbf{0}.
\end{equation*}
Since the matrix $[\mathbf{u}_1,\mathbf{u}_2,\cdots,\mathbf{u}_s]$ is an $s\times s$ matrix with full column rank, this formula~\eqref{eq:s} can be simplified as $a_{\ell j_1}+a_{\ell j_2}+a_{\ell j_3}+a_{\ell j_4}=0$ for all $\ell \in[s]$. On the other hand, by the same argument as in the proof
of the case (i), $\sum^t_{\ell=1}(a_{(s+\ell) j_1}+a_{(s+\ell) j_2})\mathbf{e}^{i_1}_{\ell}=\mathbf{0}$ and $\sum^t_{\ell=1}(a_{(s+\ell) j_3}+a_{(s+\ell) j_4})\mathbf{e}^{i_2}_{\ell}=\mathbf{0}$ can be obtained, which forces $a_{(s+\ell) j_1}+a_{(s+\ell) j_2}=0$ and $a_{(s+\ell) j_3}+a_{(s+\ell) j_4}=0$ for all $\ell\in[t]$. Hence, it is clear that $\sum^{4}_{z=1}a_{\ell j_z}=0$ for all $\ell\in[s+t]$.
By definition, any four columns of the desired matrix $A$ are linearly independent over $\mathbb{F}_2$. This contradiction completes the proof of the lemma.
\end{proof}

With the above preparation, we give the following construction of binary linear LRCs with minimum distance $d\geq6$.

\begin{construction} \label{cons-lrc}
Let $b$, $s$ ,$t$ and $m$ be integers such that $s+t=2b$ and $m\geq4b$. Choose a desired matrix $A$ with size  $2b\times 2^b$ and  a $(2b-s)$-spread $\{W_1,W_2,\cdots,W_\ell\}$ of $V_m$ with size $\ell=\frac{2^m-1}{2^{2b-s}-1}$, where $b>s\geq0$ and $(2b-s)|m$. The submatrix $H^{i}_G$ is given by $H^{i}_G=\left(\mathbf{0}_{(s + m)\times1},~ G_{M_i}\cdot A\right)$ for each $i\in[\ell]$. Then the linear code $\mathcal{C}$ is constructed by parity check matrix $H$ given in~\eqref{eq:pc} with submatrices $H^{i}_G$.
\end{construction}

  Note that a $(2b-s)$-spread $\{W_1,W_2,\cdots,W_\ell\}$ of $V_m$  exists if and only if $(2b-s)$ divides $m$. By Lemma~\ref{lam:1},  it is known that the size of   a $(2b-s)$-spread is $\ell=\frac{2^m-1}{2^{2b-s}-1}$. Additionally, for the existence of the desired matrix $A$, it is required that $b>s\geq0$.
 The submatrix $A_2$ of $A$ can be viewed as  a parity check matrix with parameters $[2^b,2^b-2b-s,3]$. According to the Griesmer bound,
the parameters of this code should satisfy
$
 2^b\geq \sum_{i=0}^{2^b-2b-s}\left\lceil\frac{3}{2^i}\right\rceil,
$
 which follows from $0\leq s<b$.

Henceforth, we will consider a binary linear code obtained from Construction~\ref{cons-lrc}. We have the following theorem.
\begin{theorem}\label{Thm:paras}
Let $b$ be an integer such that $b\geq3$.
The code $\mathcal{C}$ constructed by the parity check matrix $H $ from Construction \ref{cons-lrc} is an~ $[n=(r+1)\ell, k=n-s- \ell-m, d\geq6]$ binary linear LRC with locality $r=2^b$, which is $k$-optimal and attains the bound~\eqref{eq:3}.
\end{theorem}
\begin{proof} This proof consists of two parts. In Part $1$ we will prove that the dimension $k$ of the corresponding codes achieves the bound~\eqref{eq:3}, i.e., $k=n-s- \ell-m$. As to Part $2$, we will show that the minimum distance $d\geq6$.

Part $1$: It is easy to determine the parameter $n=(r+1)\ell$, $k\geq \frac{rn}{(r+1)}-s-m$ and $r=2^b$ by the parity check matrix $H$ in Construction \ref{cons-lrc}, where $\ell=\frac{2^m-1}{2^{2b-s}-1}$. Clearly,
\begin{equation*}
 \min \left\{\log _{2}\left(1+\frac{r n}{2}\right), \frac{r n}{(r+1)(r+2)}\right\}=\log _{2}\left(1+\frac{r n}{2}\right)
\end{equation*} can be obtained when $m\geq4b$ and $\ell=\frac{2^m-1}{2^{2b-s}-1}$.
 By the bound~\eqref{eq:3},
 \begin{equation*}
 k\leq\frac{rn}{r+1}-\left\lceil\log_2(1+\frac{rn}{2})\right\rceil=\frac{rn}{r+1}-\left\lceil\log_2(1+2^{b-1}(2^b+1)\ell)\right\rceil=\frac{rn}{r+1}-m-s \end{equation*}
  as a result of  $2^{m+s-1}<1+2^{b-1}(2^b+1)\ell\leq2^{m+s}$.
  Hence, we obtain $k= n-\ell-s-m$.

Part $2$:  To verify that the code $\mathcal{C}$ has the minimum distance $d\geq6$ is equivalent to showing that the submatrix $H^i_G$ of $H$ satisfies the conditions in Lemma \ref{lem:condition}. Notably, the sum of the first $\ell$ rows of $H$ is an all-one vector. Thus the minimum distance of $\mathcal{C}$ must be even.

Case (i): For any two columns $\mathbf{h}^{i}_{j_1}$ and $\mathbf{h}^{i}_{j_2}$ vectors from $H^i_G$, it is obvious that $\mathbf{h}^{i}_{j_1}+\mathbf{h}^{i}_{j_2}\neq\mathbf{0}$. Therefore, $H^i_G$ satisfies condition (1) in Lemma~\ref{lem:condition}.

Case (ii): Consider the four columns $\{\mathbf{h}^{i_z}_{j_z}\}^4_{z=1}$ from $H$ belong to the same block, i.e., $i_1=i_2=i_3=i_4$.
Without loss of generality, we assume in the contradiction method that $\mathbf{h}^{i_1}_{j_1}+\mathbf{h}^{i_1}_{j_2}+\mathbf{h}^{i_1}_{j_3}+\mathbf{h}^{i_1}_{j_4}=\mathbf{0}$. Similar to the proof of Lemma \ref{lem:independet},  $\sum^s_{\ell=1}(\sum^4_{z=1}a_{\ell j_z})\mathbf{u}_\ell=\mathbf{0}$ and  $\sum^t_{\ell=1}(\sum^4_{z=1}a_{(s+\ell) j_z})\mathbf{e}^{i_1}_\ell=\mathbf{0}$.
Since $\mathbf{u}_1, \mathbf{u}_2, \cdots,\mathbf{u}_s,\mathbf{e}^{i_1}_{s+1},\cdots,\mathbf{e}^{i_1}_{s+t}$ are linearly independent, $\sum^4_{z=1}a_{\ell j_z}=0$ for all $\ell\in[s+t]$, which implies that $\mathbf{a}_{j_1},\mathbf{a}_{j_2},\mathbf{a}_{j_3},\mathbf{a}_{j_4}$ from $A$ are linearly dependent. This result contradicts with the definition of $A$. Thus for any four columns $\{\mathbf{h}^{i_z}_{j_z}\}^4_{z=1}$, $\mathbf{h}^{i_1}_{j_1}+\mathbf{h}^{i_1}_{j_2}+\mathbf{h}^{i_1}_{j_3}+\mathbf{h}^{i_1}_{j_4}\neq\mathbf{0}$.
In particular, if $\mathbf{0}\in \{\mathbf{h}^{i_z}_{j_z}\}^4_{z=1}$,  $\sum^4_{z=2}a_{\ell j_z}=0$ also hold for all
$\ell\in[s+t]$. By the definition of $A$, $\sum^4_{z=1}\mathbf{h}^{i_z}_{j_z}\neq\mathbf{0}$. Hence, $H^i_G$ satisfies condition (2) in Lemma~\ref{lem:condition}.


Case (iii): Two of $\{\mathbf{h}^{iz}_{jz}\}^4_{z=1}$ belong to one block and the other two lie in a different block. Then their sum is not equal to zero by Lemma \ref{lem:independet}, proving that $H^i_G$ satisfies the condition (3) in Lemma~\ref{lem:condition}.

 As a consequence, the lower part $H_G$ of $H$ satisfies three conditions in Lemma \ref{lem:condition}. This completes the proof of Theorem~\ref{Thm:paras}.
 \end{proof}

Next, we give two examples to illustrate the corresponding construction in
detail. Example~\ref{ex:1}, by Theorem~\ref{Thm:paras}, shows how to construct the $k$-optimal binary linear LRC from
Construction~\ref{cons-lrc}, Example~\ref{ex:2} is a special case when $s = 0$.

 \begin{example}\label{ex:1}
 Let $b=3$, $s=2$ and $m=12$ in Construction~\ref{cons-lrc}. Let $\{W_{1},W_{2},\cdots,W_{273}\}$ be a $4$-spread of $V_{12}$. Denote a basis of $W_i$ by $\{\mathbf {e}^{i}_1,\mathbf {e}^{i}_2,\mathbf {e}^{i}_3,\mathbf {e}^{i}_4\}$ for $i \in[273]$.
Then we choose a matrix $G_{M_i}$ and a desired matrix $A_{6\times 8}$ as follows:
\begin{equation*}
G_{M_i}=\left(\begin{array}{ccccccc}
1&1&0&0&0&0\\
0&1&0&0&0&0\\
\mathbf{0}&\mathbf{0}&\mathbf{e}^i_1& \mathbf{e}^i_2& \mathbf{e}^i_3& \mathbf{e}^i_4
\end{array}\right),\quad A_{6\times 8}=
\left(\begin{array}{llllllll}
A_1\\
A_2
\end{array}\right),
\end{equation*}
where $A_1=\left(\begin{array}{llllllll}
1 &0 & 0 & 0 & 0 & 0 & 1 & 0 \\
0 &1 & 0 & 0 & 0 & 0 & 0 & 1 \\
\end{array}\right)$ and $A_2=\left(\begin{array}{llllllll}
1 &1 & 1 & 0 & 0 & 0 & 0 & 0 \\
1 &0 & 0 & 1 & 0 & 0 & 1 & 1 \\
0 &1 & 0 & 0 & 1 & 0 & 1 & 0 \\
0 &1 & 0 & 0 & 0 & 1 & 1 & 1 \\
\end{array}\right)$.
For example, let $\alpha$ be a primitive element of $\mathbb{F}_{2^{12}}$ such that $\alpha^{12}+\alpha^7 + \alpha^6 + \alpha^5 + \alpha^3 + \alpha + 1=0$. Let $\ell=\frac{2^{12}-1}{2^4-1}$ and $\gamma=\alpha^{\ell}$. We get a basis $\{\alpha^0,\alpha^0\gamma,\alpha^0\gamma^2,\alpha^0\gamma^3\}$ and a basis $\{\alpha^1,\alpha^1\gamma,\alpha^1\gamma^2,\alpha^1\gamma^3\}$ of subspace $W_1$ and subspace $W_2$ respectively.
Then we have
\begin{equation*}
G_{M_i}A=\left(\begin{array}{llllllllllllllllll}
1  &&1  &&0 & &0 & &0 & &0&  &1&  &1\\
0  &&1 &&0  &&0 & &0 & &0  &&0&  &1\\
\mathbf{e}^i_1+\mathbf{e}^i_2 & &\mathbf{e}^i_1+\mathbf{e}^i_3+\mathbf{e}^i_4 & &\mathbf{e}^i_1&& \mathbf{e}^i_2 &&\mathbf{e}^i_3 &&\mathbf{e}^i_4 && \mathbf{e}^i_2+\mathbf{e}^i_3+\mathbf{e}^i_4 & & \mathbf{e}^i_2+\mathbf{e}^i_4 \\
\end{array}\right),
\end{equation*}
where $\mathbf{e}^i$ is an element in $\mathbb{F}_{2^{12}}$. Columns of $H^i_G$ are binary
expansions of the each column vector $(\mathbf{0}, G_{M_i}\cdot A)$. For example, fixing a basis $\{\alpha^0, \alpha^1,\ldots,\alpha^{11}\}$, from the submatrix $H^1_G$, $\mathbf{e}^1_1+\mathbf{e}^1_2=\alpha^0+\gamma=\alpha^{10} + \alpha^9 + \alpha^8 + \alpha^4 + \alpha^3 + \alpha^2 + 1=(1, \alpha,\ldots,\alpha^{11}) \cdot(1,0,1,1,1,0,0,0,1,1,1,0)^T$. Thus the  binary
expansion of the vector of $\mathbf{e}^1_1+\mathbf{e}^1_2$ with respect to the basis is $(1,0,1,1,1,0,0,0,1,1,1,0)$.
Then  we obtain the matrices $H^1_G$ and $H^2_G$ as follows by expanding the column vectors of the submatrices $G_{M_i}\cdot A$ $(i=1,2)$ with respect to the bases respectively:
\begin{equation*}
 H^1_G=
\left(\begin{array}{llllllllllllllllllllllll}
 0& 1& 1& 0& 0& 0& 0& 1& 1 \\
 0& 0& 1& 0& 0& 0& 0& 0& 1 \\
 0& 1& 1& 1& 0& 1& 1& 0& 1 \\
 0& 0& 1& 0& 0& 1& 0& 1& 0 \\
0& 1& 0& 0& 1& 1& 1& 1& 0\\
0& 1& 1& 0& 1& 1& 0& 0& 1\\
 0 &1 &0 &0 &1 &1 &1 &1 &0\\
 0 &0& 0 &0& 0 &1& 1& 0& 1\\
 0& 0& 1& 0& 0 &0 &1& 1& 1\\
 0 &0 &0 &0 &0 &0 &0 &0 &0\\
 0 &1 &1 &0 &1 &0 &1 &0 &0\\
0 &1 &1 &0 &1& 1 &0 &0 &1\\
0 &1 &0 &0 &1 &0 &0 &1 &1\\
0 &0 &1 &0 &0 &1 &0 &1 &0
\end{array}\right)\quad H^{2}_G=
\left(\begin{array}{lllllllll}
0&1& 1& 0& 0& 0& 0& 1& 1 \\
0 &0 &1 &0 &0 &0 &0 &0& 1 \\
0&0 &1& 0& 1& 0& 0& 1& 0\\
0&1 &1& 1 &0& 0& 1& 0& 1 \\
0&0& 1 &0 &1 &0 &0 &1 &0 \\
0& 0& 0& 1& 0& 1& 0& 1& 1\\
0& 1 &0 &0 &1 &1 &0 &1 &1 \\
0& 0& 0 &1 &0 &1 &0 &1 &1 \\
0& 1& 1 &1 &0& 0 &0 &1& 0  \\
0& 1 &0 &1& 1& 0& 0 &0 &0  \\
0& 0 &0 &0 &0 &0 &0 &0 &0 \\
0& 0 &0 &1 &0 &1 &0 &1 &1  \\
0& 1 &0& 0 &1& 1 &0 &1 &1 \\
0& 1 &1& 0& 0& 1 &0 &0 &1
\end{array}\right)
\end{equation*}
It can be verified that any $5$ columns of $H$ in~\eqref{eq:pc} are linearly independent, so parity check matrix $H$ defines a $[2457,2170,6]$ $k$-optimal binary linear LRC with locality $r = 8$ by Theorem~\ref{Thm:paras}.

\end{example}
\begin{example}\label{ex:2}
 Taking $t=2$ and $s=0$ in the above example. Let $\left\{W_{0}, W_{1}, W_{2}, W_{3}, W_{4}\right\}$ be a $2$-spread of $V_{4}$ and $\{\mathbf{e}^i_1, \mathbf{e}^i_2\}$ be a basis of subspace $W_i$. By choosing
\begin{equation*}
  A=\left(\begin{array}{ll}
1 & 0 \\
0 & 1
\end{array}\right)\quad \mbox{and} \quad G_{M_i}=\left(\begin{array}{ccccc}
\mathbf{e}^i_1& \mathbf{e}^i_2
\end{array}\right),
\end{equation*}
we obtain a $[15, 6, 6]$  $k$-optimal binary linear LRC with locality $r=2$. Similarly, when $s=0$, more $k$-optimal binary linear LRCs are listed in the following Table~\ref{table-0824}.
    \begin{table}[ht]
  \centering
  {
  \caption{k-optimal binary linear LRCs with $d=6$}\label{table-0824}
\begin{tabular}{|c|c|c|c|c|c|c|c|c|c|c|c|c|c|}
  \hline
    $r$ & 2& 2& 2& 4& 4& 5& 6\\ \cline{1-8}
     $k$  & 6&  36& 162&  21&  210&  60&  155\\ \cline{1-8}
  $n$ & 15 &  63 & 255&  36&  292&  85&  195\\
  \hline
 \end{tabular}
     }
\end{table}
\end{example}

\begin{remark}
Note that an LRC with the same parameters as in Example~\ref{ex:2} was also constructed in  Example $2$ of \cite{Anyu_2019}. Correspondingly, taking $s=0$ in Construction \ref{cons-lrc} and Theorem \ref{Thm:paras}, we obtain an  $[n=\frac{2^{m}-1}{2^{t}-1}, k=\frac{r n}{r+1}-m, d \geq 6]$ binary linear LRC with locality $r=2^{t}$, which includes the construction of $k$-optimal binary linear LRCs in \cite{Anyu_2019}.
\end{remark}

In fact, Construction~\ref{cons-lrc} generates a family of $k$-optimal binary linear LRCs with locality $r=2^b$ when $\ell$ belongs to
a determined range. This point is presented in detail in  Theorem~\ref{tm:op-l}.
\begin{theorem}\label{tm:op-l}
Assume that $r=2^{b}$ for $b\geq3$ and let $m,s,b$ be integers such that $m \geq 4 b$ and $0\leq s< b$. If
\begin{equation*}
    \frac{2^{m+s-1}-1}{2^{b-1}(2^{b}+1)}<\ell \leq \frac{2^m-1}{2^{2b-s}-1},
\end{equation*}
there exists an $\left[n=(r+1)\ell, k=r \ell-m-s, d\geq6\right]$ binary linear LRC with locality $r=2^{b}$, which is $k$-optimal with respect to the bound~\eqref{eq:3}.
\end{theorem}
\begin{proof} By Theorem~\ref{Thm:paras}, $\mathcal{C}$ is an $[n=(r+1)\ell, k\geq \frac{rn}{r+1}-m-s, d\geq6]$ binary linear LRC. Hence, we need to show that $k\leq r \ell-m-s$. Due to the condition that
\begin{equation*}
\frac{2^{m+s-1}-1}{2^{b-1}(2^{b}+1)}<\ell \leq \frac{2^m-1}{2^{2b-s}-1},
\end{equation*}
we have
\begin{equation}\label{eq:6} 2^{m+s-1}<1+2^{b-1}(2^{b}+1)\ell\leq\frac{2^{b-1}(2^{b}+1)(2^m-1)}{2^{2b-s}-1}+1\leq 2^{m+s}.
\end{equation}
 Furthermore, by the bound~\eqref{eq:3} and the formula~\eqref{eq:6},
\begin{eqnarray}
\nonumber
  k\leq \frac{rn}{r+1}-\left\lceil\log_2(1+\frac{rn}{2})\right\rceil&=&\frac{rn}{r+1}-  \lceil\log_2(1+\frac{r(r+1)\ell}{2})\rceil\\ \nonumber
   &=& \frac{rn}{r+1}-\lceil\log_2(1+2^{b-1}(2^b+1)\ell) \rceil\\ \nonumber
   &=& \frac{rn}{r+1}-m-s.
\end{eqnarray}
 As $k\geq\frac{rn}{r+1}-m-s$ in Theorem \ref{Thm:paras}, $k=\frac{rn}{r+1}-m-s$, which is $k$-optimal with respect to the bound~\eqref{eq:3}.
\end{proof}

Notice that a necessary condition for the existence of the $(2b-s)$-spread is $(2b-s)\mid m$ in Construction \ref{cons-lrc}. This condition restricts the parameters of LRC codes constructed using intersection subspace. For the case of $(2b-s)\nmid m$, we utilize the partial $(2b-s)$-spread of $V_m$ to replace the $(2b-s)$-spread. Although the size of a maximum partial spread of $V_m$ is not known when $(2b-s)\nmid m$, an explicit construction for a partial $t$-spread of size $\frac{q^m-q^{(2b-s)}(q^z-1)-1}{q^{(2b-s)}-1}$ is presented in \cite{Etzion2011}, where $z\equiv m \bmod (2b-s)$. Hence, we obtain the following theorem.
\begin{theorem}\label{th3:parity}
  Let $m\geq4b$. There exists an $k$-optimal binary linear LRC with parameters $[n=(r+1)\ell, k=n-\ell-s-m,d\geq6]$ and locality $r=2^b$ if there exists a partial $(2b-s)$-spread of $V_m$ for $(2b-s)\nmid m$, where $\ell=\frac{2^{m}-2^{(2b-s)}(2^z-1)-1}{2^{(2b-s)}-1}$, $0\leq s< b$ and $z\equiv m\bmod(2b-s)\leq b$.
 \end{theorem}
 \begin{proof}
 By the method analogous to that used in the proof
of Theorem \ref{Thm:paras}, an LRC code has parameters $n=(r+1)\ell$, $k\geq n-\ell-m-s$, $d\geq6$. Hence, we only need to show that its dimension $k$ satisfies the bound~\eqref{eq:3}:
\begin{equation*}
k \leq \frac{r n}{r+1}-\min \left\{\log _{2}(1+\frac{r n}{2}), \frac{r n}{(r+1)(r+2)}\right\},
\end{equation*}
which is equivalent that $k\leq n-\ell-s-m$, i.e.,

\begin{equation}\label{eq::1}
 m+s-1<\min \left\{\log _{2}\left(1+\frac{r n}{2}\right), \frac{r n}{(r+1)(r+2)}\right\}\leq m+s. \end{equation}

As $n>5(r+1)(r+2)$, inequality~\eqref{eq::1} can be written as
\begin{equation}\label{eq:8}
2^{m+s-1}\leq 1+\frac{r n}{2}\leq2^{m+s}.
\end{equation}
For the left side of the inequality, it has to be verified that  $(2^{m+s-1}-1)(2^{2b-s}-1) < 2^{b-1}(2^b+1)(2^{m}-2^{2b-s}(2^z-1)-1)$. Since $m\geq4b$, $0\leq z\leq b$ and $0\leq s< b$, $2^{m+b-1}+2^{4b-s-1} =2^{m-4b}\cdot2^{3b-1}+2^{4b-s-1}>2^{4b+z-s-1}+2^{3b-1-s+z}$.  Then $(2^{m+s-1}-1)(2^{2b-s}-1) <2^{b-1}(2^b+1)(2^{m}-2^{2b-s}(2^z-1)-1)$ follows from $2^{m+s-1}\geq2^{2b-1}$ and $2^{2b-s}>2^{b-1}$, which supports the left side of formula~\eqref{eq:8}.

%

The right side of the inequality~\eqref{eq:8} holds by using similar arguments as in the above paragraph. The proof has been completed.
\end{proof}

Similar to the above analysis of Theorem~\ref{tm:op-l}, we obtain a family of $k$-optimal LRCs with locality $r=2^b$ from Theorem~\ref{th3:parity} when $\ell$ lies within a specific range.
 \begin{theorem}\label{th4:l}
 Let $r$, $b$ and $s$ be integers such that $r=2^b$, $b\geq3$ and $0\leq s< b$. Suppose that $m \geq 4b$ is an integer and $1\leq z\equiv m\bmod(2b-s)\leq b$. When
\begin{equation}\label{th4:eq}
\frac{2^{m+s-1}-1}{2^{b-1}(2^{b}+1)}<\ell \leq \frac{2^{m}-2^{(2b-s)}(2^z-1)-1}{2^{(2b-s)}-1},
\end{equation}
the code $\mathcal{C}$ in Theorem~\ref{th3:parity} is a $k$-optimal binary linear LRC with parameters  $[n=(r+1)\ell, k=n-\ell-m-s, d\geq6]$.
 \end{theorem}

 \begin{proof} It is easy to construct the code $\mathcal{C}$ with parameters $n=(r+1)\ell, d\geq6$ and $k\geq n-\ell-m-s$ by using the partial $(2b-s)$-spread in Construction \ref{cons-lrc},  where $\ell=\frac{2^{m}-2^{(2b-s)}(2^z-1)-1}{2^{(2b-s)}-1}$. In the similar way provided in Theorem~\ref{tm:op-l}, we prove $k\leq n-\ell-m-s$ below.

 Combining the proof of Theorem~\ref{th3:parity} with the equation~\eqref{eq:8} and the condition~\eqref{th4:eq}, we derive the following chain of inequalities :
 \begin{equation*} 2^{m+s-1}<1+2^{b-1}(2^{b}+1)\ell\leq\frac{2^{b-1}(2^{b}+1)(2^{m}-2^{(2b-s)}(2^z-1)-1)}{2^{(2b-s)}-1}+1\leq 2^{m+s},
\end{equation*}
  which implies $k\leq n-\ell-m-s$  by the bound~\eqref{eq:3}. Therefore, $k=n-\ell-m-s$ proves the theorem.
\end{proof}

An example of Theorem~\ref{th3:parity} and Theorem~\ref{th4:l} is presented below.

\begin{example}\label{ex:3}
Let $m=12$, $b=3$, $s=1$ and let $\left\{W_{1}, W_{2}, \ldots, W_{129}\right\}$ be a partial $5$-spread of $V_{12}$. Then there exists an $[n=1161, k=1019, d \geq 6]$ binary LRC with locality $r=8$ by Theorem~\ref{th3:parity} . This code is $k$-optimal since it attains the bound~\eqref{eq:3}. Moreover, taking $113\leq \ell\leq 129$, the code is a $k$-optimal binary linear LRC with parameters $[n=(r+1)\ell, k= n-\ell-m-s, d\geq6]$ by Theorem~\ref{th4:l}.

\end{example}

Here we list parameters of $k$-optimal binary linear LRCs with disjoint local repair groups given by Theorem \ref{Thm:paras} and Theorem \ref{th3:parity} in Table \ref{table-1} for $3\leq b\leq6$ and  $12\leq m\leq 40$ , which achieve the maximum value obtained from the bound~\eqref{eq:3}. The values highlighted in bold in Table~\ref{table-1} are new parameters of $k$-optimal binary linear LRCs in the current paper. The parameters of LRCs with the same locality $r$ in~\cite{Anyu_2019} are also listed in Table~\ref{table-1}.
\begin{table*}[ht]
  \centering
  {
  \caption{$k$-optimal binary linear LRCs with $d\geq6$}\label{table-1}
  \footnotesize\begin{tabular}{|c|c|c|c|}
    \hline
     \multicolumn{1}{|c|}{} &$[n,k;k/n]$ from Theorem \ref{Thm:paras}  & $[n,k;k/n]$ from Theorem \ref{th3:parity}& $[n,k;k/n]$ in \cite{Anyu_2019}  \\ \cline{1-4}
     \multirow{2}{*}{$r=8$}&  \multirow{2}{*}{\boldsymbol{$[2457,2170; 0.8832]$}}  &{\boldsymbol{$[1161,1019; 0.8777]$} }  & \multirow{2}{*}{$[585,508; 0.8684]$ } \\
     &  &{\boldsymbol{$[9801, 8696; 0.8873]$ }}&\\
   &   \multirow{2}{*}{\boldsymbol{$[10066329,8947822; 0.8889]$}}  &{ \boldsymbol{$[38025, 33782; 0.8884]$ }} &\multirow{2}{*}{{$[2396745, 2130416;  0.8889]$}}  \\
     &   &{\boldsymbol{ $[1258281, 1118449;  0.8889]$}}&  \\
     \cline{1-4}
         \multirow{4}{*}{$r=16$} &  \multirow{2}{*}{\boldsymbol{$[4527185, 4260854; 0.9412]$}} & {\boldsymbol{$[4369, 4096; 0.9375]$}} & \multirow{2}{*}{{$[1118481, 1052664;0.9412]$}} \\
      & &  {\boldsymbol{$[1122833, 1056760; 0.9412]$}} &\\
        &\multirow{2}{*}{\boldsymbol{$[602957989425, 567489872357;  0.9412]$}}

        &  {\boldsymbol{$[1150033, 1082360; 0.9412]$}}  & \multirow{2}{*}{{$[73300775185, 68988964840;0.9412]$}} \\
      &  & {\boldsymbol{ $[287458321,270548976; 0.9412]$}}&\\   \cline{1-4}
       \multirow{2}{*}{$r=32$} & \multirow{2}{*}{\boldsymbol{$[562436193, 545392638; 0.9697]$}} & {\boldsymbol{$[33825,32780;0.9691]$}} & \multirow{2}{*}{{$[34636833, 33587202;0.9697]$}} \\
      & &  {\boldsymbol{$[34670625,33619970; 0.9697]$}} & \\   \cline{2-4}
     \cline{1-4}
     \multirow{2}{*}{$r=64$} & \multirow{2}{*}{\boldsymbol{$[4276545, 4210724; 0.9846]$}} & {\boldsymbol{$[68290625, 67239968; 0.9846]$} }& \multirow{2}{*}{{$[266305, 262184;0.9845]$}} \\
      &  & {\boldsymbol{$[1091051585,1074266140; 0.9846]$} }& \\   \cline{2-4}
     \cline{1-4}
  \end{tabular}
     }
\end{table*}
\begin{remark}

In \cite{Anyu_2019}, Wang et al. constructed the parity check matrix of binary linear LRCs based on a $2^{2b}$-ary Hamming code with length $\frac{2^m-1}{2^{2b}-1}$. Then they obtained an $[n'=\frac{2^{m}-1}{2^{b}-1}, k'=\frac{r n'}{r+1}-m, d \geq 6]$ binary linear LRC with disjoint local
repair groups and locality $r=2^b$. Furthermore, their code rate is $\frac{k'}{n'}=\frac{r}{r+1}-\frac{m}{n'}$. Compared with the code rate of binary linear LRCs for $r=2^b$ in \cite{Anyu_2019}, our constructions have larger code rate. In this paper, taking $0< s<b$,  the length of $k$-optimal LRC is $n=(2^b+1)\frac{2^m-1}{2^{2b-s}-1}$, which is approximately $2^s$ times greater than $n'$, and the dimension $k$ is $\frac{rn}{r+1}-s-m$. Hence, for the same $b$ and $r$, it is easy to show that the code rate $\frac{k}{n}=\frac{r}{r+1}-\frac{s}{n}-\frac{m}{n}$ is larger than $\frac{k'}{n'}$ because $\frac{s+m}{n}<\frac{m}{n'}$.
\end{remark}

Table \ref{table-2} gives the summary of $k$-optimal binary linear LRCs with disjoint local repair groups whose minimum distance $d\geq6$. We also list results of Theorem~\ref{tm:op-l} and Theorem~\ref{th4:l}. The comparison of the number of the disjoint local repair groups illustrates that $k$-optimal binary linear LRCs with a wider range of parameters can be obtained from Theorem~\ref{tm:op-l}.
Here, $\mu_{2}(m, 2b)$ denotes the size of a maximum partial $2b$-spread in $V_m$.
{\small
\begin{table*}[ht]
  \centering
  \caption{$[n=(r+1)\ell,k,d \geq 6]$ $k$-optimal  binary linear LRCs with respect to the bound \eqref{eq:3}}\label{table-2}
\begin{tabular}{|c|c|c|c|}
\hline Ref. & $r$ & The number of repair groups & Conditions \\ \hline
\cite{goparaju2014binary} & 2 & $\ell=\frac{2^m-1}{3}$ & $2|m,m\geq6$ \\ \hline
\cite{kim2019new}  &3& $\frac{2^{m}-1}{6}\leq \ell< \frac{2^{m}-1}{3}$ &   $m\geq6$ \\ \hline
 \cite{Anyu_2019}& $2^{b}$ & $\ell=\frac{2^{m}-1}{2^{2 b}-1}$ & $2 b \mid m, m \geq 4 b$ \\ \hline
  \cite{Jingxue2019}& $2^{b}$ & $\lfloor\frac{2^{m-1}-1}{2^{b-1}(2^{b}+1)}\rfloor
  +1\leq\ell \leq \mu_{2}(m, 2b)$  & $m \geq 4 b$ \\ \hline
 Thm~\ref{tm:op-l} & $2^{b}$ & $
\frac{2^{m+s-1}-1}{2^{b-1}(2^{b}+1)}<\ell \leq \frac{2^m-1}{2^{2b-s}-1}$ & $(2b-s)|m,m \geq 4 b,0\leq s<b$ \\ \hline
 \multirow{2}{*}{Thm~\ref{th4:l}}& \multirow{2}{*}{$2^{b}$} & \multirow{2}{*}{$
\frac{2^{m+s-1}-1}{2^{b-1}(2^{b}+1)}<\ell \leq \frac{2^{m}-2^{(2b-s)}(2^z-1)-1}{2^{(2b-s)}-1}$} & \multirow{1}{*}{$(2b-s)\nmid m,m \geq 4 b,0\leq s<b$}\\
&&&  $1\leq z\equiv m\bmod(2b-s)\leq b$ \\ \hline
\end{tabular}
\end{table*}
}
\begin{remark}
As a comparison, the $k$-optimal binary linear LRCs generated by this paper have more flexible parameters $[n,k]$ than those in \cite{Anyu_2019},\cite{Jingxue2019} for a fixed locality $r=2^b$. Particularly, if we take $s=0$ in Construction \ref{cons-lrc}, we have $\frac{2^{m-1}-1}{2^{b-1}(2^{b}+1)}<\ell \leq \frac{2^m-1}{2^{2b}-1}$ in Theorem~\ref{tm:op-l}. Note that $\ell=\frac{2^{m}-1}{2^{2 b}-1}$ in~\cite{Anyu_2019} and $\mu_{2}(m, 2b)\leq \frac{2^m-1}{2^{2b}-1}$ in~ \cite{Jingxue2019}. Hence, $k$-optimal binary linear LRCs in \cite{Anyu_2019} and \cite{Jingxue2019} are included in our construction. For example, let $b=3$ and $m=12$, we obtain $\ell=65$ in \cite{Anyu_2019} and $56<\ell\leq65$ in \cite{Jingxue2019}. However, Theorem~\ref{tm:op-l} and Theorem~\ref{th4:l} yield $56<\ell\leq65$  and $227<\ell\leq 273$ respectively, which shows that our method constructs more $k$-optimal binary linear LRCs with the same locality.
\end{remark}
More specially, we concentrate on the value of $\ell$, then we have the following corollary.
\begin{corollary}
Let $S=\bigcup_{ m\geq 4b,0\leq s<b \atop
1 \leq z\equiv m\bmod(2b-s)\leq b
}\left(\left[
\frac{2^{m+s-1}-1}{2^{b-1}(2^{b}+1)}+1, \frac{2^{m}-2^{(2b-s)}(2^z-1)-1}{2^{(2b-s)}-1}\right]\right)$. Suppose that $\ell\in S$, then there exists an $[n = (r + 1)\ell,k = r -m -s,d \geq6]$ binary linear LRC with locality $r = 2^b$, which is $k$-optimal with respect to the bound~\eqref{eq:3}.
\end{corollary}

\begin{proof}
Combining Theorem~\ref{tm:op-l} with Theorem~\ref{th4:l},
this corollary  can be obtained directly.
\end{proof}

\section{Shortening LRC}\label{sec:5}
The shortening technique can be applied to the derivation of binary linear LRCs with new parameters.  Let $\mathcal{C}$ be an $[n,k,d]$ code over $\mathbb{F}_q$ and let $\mathcal{S}$ be any set of $i\in[n]$ coordinates. Consider the set $\C(\mathcal{S})$ of codewords which are $0$ on $\mathcal{S}$; this set is a subcode of $\mathcal{C}$. Deleting the same coordinate $i$ for all $i\in\mathcal{S}$ in each codeword of $\mathcal{C}(\mathcal{S})$ gives
a code over $\mathbb{F}_q$ of length $n-|\mathcal{S}|$ called the code shortened on $\mathcal{S}$ and denoted $\mathcal{C}'$. Hence, we obtain the following theorem with respect to the shortened LRCs.

%

\begin{theorem}\label{the:shorten}
Let $\mathcal{C}$ be an $[n,k,d]$ $k$-optimal binary linear LRC constructed in Theorem~\ref{Thm:paras} or Theorem~\ref{th3:parity} such that $n \geq 2(r+1)$ and $k \geq 2r $.
\begin{itemize}
  \item [(1)]  Suppose that $a$ is an integer that satisfies $0\leq a\leq\frac{n}{r+1}$. An $[n', k', d']$ LRC $\mathcal{C}'$ with locality $r$ can be obtained by shortening $\mathcal{C}$, where parameters of $\C'$ satisfy $n'=n-a(r+1), k' \geq k-ar$ and $d'\geq d$.
  \item [(2)] Let $H^{i}=\left(\begin{array}{l}
H_i\\
H^i_G
\end{array}\right)$ for all $i\in[\ell]$. Removing a column of each distinct submatrix $H^{i_1},H^{i_2},\cdots,H^{i_\tau}$ from the parity check matrix $H$ respectively for $i_\tau\in[\ell]$, then there exists a shorten LRC with parameters $[n'=n-\tau,k'= k-\tau,d'=d]$.
\end{itemize}
\end{theorem}
\begin{proof}
(1) Assume that $H$ is a parity check matrix of $\mathcal{C}$. The first $\ell$
rows $\mathbf{h}_1, \cdots, \mathbf{h}_{\ell}$ from $H$ form a set of locality rows of $\mathcal{C}$, where $\mathbf{h}_i\in V_n$ with $|\operatorname{supp}(\mathbf{h}_i)|=r+1$. Consider the first $\tau$ locality row of $H$, where $1\leq \tau\leq \ell $. By deleting the first $\tau$ locality rows  $\mathbf{h}_1,\mathbf{h}_2,\cdots,\mathbf{h}_{\tau}$ and the corresponding column whose index belongs to the support of $\mathbf{h}_i$ for all $i\in[\tau]$, we obtain an $m'\times n'$ submatrix $H'$ with $n'=n-\tau(r+1)$, $m'= n-k-\tau$. Let $\mathcal{C}'$ be the $[n',k',d']$ linear code with the parity check matrix $H'$. Due to rank$(H')\leq (n'-m')$, $k'\geq k-\tau r$. Note that $\mathcal{C}'$ is a shortening code of $\mathcal{C}$, then $\mathcal{C}'$ is an LRC code with minimum distance $d'\geq d$. This completes the proof.

(2) Since each submatrix $H^i$ is generated by a desired matrix $A$ and a matrix $G_{M_i}$ for $i\in[\ell]$. Note that $A$ can be viewed as a parity check matrix of a linear code with minimum distance $d=5$. Assume that $A'$ is an $[r-1,r-1-(s+t),d]$ matrix obtained by deleting a column of $A$. Then we construct $H^{i_1},\cdots,H^{i_\tau}$ of $H_G$ by utilizing the matrix $A'$ and the remaining $(\ell-\tau)$ submatrices of $H_G$ by utilizing the matrix $A$ in Construction \ref{cons-lrc}, where $\tau\in[\ell]$.
 Hence, a linear LRC $\mathcal{C}'$ with parameters $[n'=n-\tau, k'= k-\tau, d'=d]$ can be obtained. In particular, when $\tau=\ell$, the locality of LRC is $r-1$; otherwise, the locality of LRC is $r$.
\end{proof}
 Below, an example is given to show a shortened LRC in Theorem~\ref{the:shorten}.
\begin{example} The matrix $A'$ is generated by removing the first column from $A$ in Example \ref{ex:1}, i.e.,
$$
A'=
\left(\begin{array}{llllllll}
0 & 0 & 0 & 0 & 0 & 1 & 0 \\
1 & 0 & 0 & 0 & 0 & 0 & 1 \\
1 & 1 & 0 & 0 & 0 & 0 & 0 \\
0 & 0 & 1 & 0 & 0 & 1 & 1 \\
1 & 0 & 0 & 1 & 0 & 1 & 0 \\
1 & 0 & 0 & 0 & 1 & 1 & 1 \\
\end{array}\right).
$$
 Then the submatrix $H^i_G$ is constructed by $(\mathbf{0},~~G_{M_i}\cdot A')$ for each $i\in[2]$, and the remaining submatrix $H^i_G$ is constructed by $(\mathbf{0},~~G_{M_i}\cdot A)$ for $i\in\{3,4,\cdots,\ell\}$. Thus we obtain a $[2455,2168,6]$ binary linear LRC in Theorem~\ref{the:shorten}, which is $k$-optimal with respect to the bound~\eqref{eq:3}.

\end{example}


\section{Conclusion}\label{sec:6}

In this paper, we present an explicit construction of $k$-optimal binary linear LRCs with minimum distance $d\geq6$ by investigating parity check
matrices. In general, $k$-optimal binary linear LRCs with minimum distance $d\geq6$ and locality $r=2^b$ are constructed by $t$-spread of an $m$-dimensional vector space over $\mathbb{F}_2$ which is a collection of $t$-dimensional subspaces with pairwise trivial. Of interest is the idea of using intersection subspaces to replace the method of $t$-spread. Based on this new idea, we efficiently enlarge the range of new parameters of $k$-optimal binary linear LRCs with minimum distance $d\geq6$ and locality $r=2^b$. In fact, it yields more repair groups such that the corresponding constructions have more flexible lengths and dimensions. Compared with the previous works in \cite{Anyu_2019} and  \cite{Jingxue2019} with the same locality, the code lengths of our work are larger and the code rates are higher.

\Acknowledgements{We would like to thank Professor Chaoping Xing  for
introducing us to this problem. During this work, he provided many valuable discussions and expert advices, which are very useful for improving the quality of this
paper. The authors also would like to express their sincere gratefulness
to editor and the four anonymous reviewers  for their efforts in reviewing this article and constructive
comments.
\\
This work was supported in part by National Key R\&D Program of China  (Grant Nos. 2022YFA1004900, 2022YFA1005000), and National Natural Science Foundation of China (Grant No.  62272303).}



\begin{appendix}
\section{Examples of desired matrix $A$}\label{sec:8}
In this section, we provide an approach to construct the desired matrix $A$ with the help of a computer search program. Note that a desired matrix $A$ can be viewed as the parity check matrix of a $[2^b,2^b-2b,d\geq5]$ binary linear code. Although in \cite{Anyu_2019}, Wang et al. presented the explicit construction of binary linear code with parameters $[2^b,2^b-2b,d\geq5]$ from a shortened nonprimitive cyclic code, which can not be used directly here. A necessary condition for the desired matrix $A$ is that it contains a submatrix in which any two distinct columns are linearly independent. This makes it difficult to give an explicit construction of the desired matrix $A$.  In addition,  LRCs constructed by using an arbitrary $t\times n$ matrix $A$ have the same code length, dimension and minimum distance, which implies that we can weaken its explicit construction.
By the computer program MAGMA, we have found some examples of the desired matrix $A$, which also shows the existence of these desired matrices.
However, how to construct more desired matrices $A$ by using theoretical analysis and an effective search algorithm, remains an open problem.

We briefly recall the construction of a binary linear code with parameters $[2^b,2^b-2b,\geq5]$ in~\cite{Anyu_2019}. Let $n = 2^b +1$ and let $\alpha$ be a primitive root of $x^n-1$ with minimal polynomial $M_{\alpha}(x)$. Clearly, the degree of $M_{\alpha}(x)$ is $2b$. Define $\mathcal{C}$ to be the binary cyclic code of length $n$ generated by $(x-1)M_{\alpha}(x)$. It is not hard to show that $\left\{\alpha^{i}: i=-2,-1,0,1,2\right\}$ forms a subset of the roots of the generator polynomial of $\mathcal{C}$, so  $\mathcal{C}$ is an $\left[n=2^b+1,k=2^b-2b,d\geq6\right]$  binary linear code. The code $\mathcal{C}$ can be punctured by deleting one of the check bits to yield a code $\mathcal{C}'$ of length $2^{b}$ with $2 b$ check bits and $d \geq 5$. Hence, a $[2^b,2b]$ parity check matrix $A'$ of $\C'$ can be obtained. Notice that a desired matrix $A$ has the same parameters as matrix $A'$ because $A$ also can be viewed as a parity check matrix of an $[n=2^b,k=2^b-2b,d\geq5]$ linear code.
Hence, applying the row transformation to $A'$, $A'$ can be transformed into a desired matrix $A$.
 Here, if a $z\times n$ submatrix of a $k\times n$ desired matrix satisfies that any two distinct column vectors from the submatrix are linearly independent, we denote this desired matrix $A$ by an $[n,k]_z$ matrix. See the following examples of the desired matrix $A$.

\begin{itemize}
\item $[8,6]_4$ matrix $A$:
$$\left(\begin{array}{llllllll}
10000010 \\
01000  0 0 1 \\
1 1  1  0  0  0  0  0 \\
1 0  0  1  0  0  1  1 \\
0 1  0  0  1  0  1  0 \\
0 1  0  0  0  1  1  1 \\
\end{array}\right)$$
\item $[16,8]_5$ matrix $A$:
$$\left(\begin{array}{ccccccccccccccccccc}
0010010001101011\\
001000001000111 1\\
0 1 0 0 0 00 10 0 0 1 1 1 1 0\\
1 0 0 0 0 0 0 0 0 1 0 0 1 1 1 0\\
0 1 0 1 0 0 0 0 1 1 1 1 1 1 0 0\\
0 1 1 0 1 0 0 0 0 1 0 1 1 0 0 1\\
0 1 1 0 0 1 1 0 0 0 1 1 1 1 1 0\\
0 0 0 0 0 1 0 1 1 1 0 1 1 1 0 1
\end{array}\right)$$
\item $[32,10]_6$ matrix $A$:
$$
\left(\begin{array}{cccccccccccccccccccccccccccccccc}
10000 00000 10010 11101 01101 01110 10\\
01000 00000 01001 01110 10110 10111 01\\
00100 00000 11000 00001 00010 11111 01\\
00010 00000 10000 10110 11000 11011 01\\
10001 01000 10001 11010 01100 10111 01\\
01001 00101 10011 10100 01011 10011 01\\
00100 00100 01111 01010 01001 01011 11\\
00010 00010 11110 10001 00000 00110 00\\
00000 11000 00001 00010 11111 01000 10\\
00000 00010 01011 10101 10101 11010 01
\end{array}\right)
$$
\item $[64,12]_8$ matrix $A$:
$$
\left(\begin{array}{cccccccccccccccccccccccccccccccc}
1 0 0 0 0 0 0 0 0 0 0 0 1 0 1 1 0 0 1 1 1 0 1 0 0 1 0 0 1 1 0 1 0 1 0 0 0 1 0 0 1 0 0 0 1 0 1 0 1 1 0 0 1 0 0 1 0 1 1 1
    0 0 1 1\\
0 1 0 0 0 0 0 0 0 0 0 0 1 0 0 0 0 1 0 0 1 1 1 0 1 0 1 0 1 1 1 1 0 0 1 0 0 1 0 1 0 1 0 0 1 0 0 1 1 1 1 0 1 0 1 0 1 1 1 0
    0 1 0 0\\
0 0 1 0 0 0 0 0 0 0 0 0 0 1 0 0 0 0 1 0 0 1 1 1 0 1 0 1 0 1 1 1 1 0 0 1 0 0 1 0 1 0 1 0 0 1 0 0 1 1 1 1 0 1 0 1 0 1 1 1
    0 0 1 0\\
0 0 0 1 0 0 0 0 0 0 0 0 0 0 1 0 0 0 0 1 0 0 1 1 1 0 1 0 1 0 1 1 1 1 0 0 1 0 0 1 0 1 0 1 0 0 1 0 0 1 1 1 1 0 1 0 1 0 1 1
    1 0 0 1\\
1 0 0 0 0 0 0 1 1 0 0 0 0 1 1 0 0 0 0 0 0 1 0 1 0 1 0 1 1 0 1 1 1 0 1 1 1 0 0 1 0 0 1 0 1 0 0 1 0 0 1 1 1 0 1 1 1 0 1 1
    0 1 0 1\\
0 1 0 0 0 0 0 0 1 0 1 0 0 0 0 0 0 1 0 0 1 1 0 1 0 1 1 0 1 0 0 0 0 1 1 0 0 0 1 1 1 1 1 1 1 1 0 0 0 1 1 0 0 0 0 1 0 1 1 0
    1 0 1 1\\
0 0 1 0 0 0 0 0 1 0 0 1 1 0 1 1 0 1 1 0 1 0 1 0 1 0 1 1 0 1 1 0 1 1 0 0 1 0 0 0 0 0 1 0 0 0 1 1 0 1 0 0 0 1 1 1 1 0 0 0
    1 0 1 1\\
0 0 0 1 0 0 0 1 1 0 1 0 0 0 1 1 1 1 0 0 0 1 0 1 1 0 0 0 1 0 0 0 0 0 1 0 0 1 1 0 1 1 0 1 1 0 1 0 1 0 1 0 1 1 0 1 1 0 1 1
    0 0 1 0\\
0 0 0 0 1 0 0 1 0 0 1 1 0 0 1 0 0 1 0 0 0 0 0 1 1 0 0 1 1 1 1 0 1 1 0 1 0 1 1 0 1 0 1 0 1 0 1 0 1 1 0 1 0 1 1 0 1 1 1 1
    0 0 1 1\\
0 0 0 0 0 1 0 1 0 0 1 0 1 1 1 0 1 0 0 0 0 0 1 0 0 1 1 1 0 1 1 0 0 0 0 0 1 1 0 1 1 1 0 0 1 0 0 0 0 0 1 0 1 1 1 0 1 0 0 1
    0 1 0 0\\
0 0 0 0 0 0 1 1 0 0 1 0 1 1 0 1 0 0 1 1 0 0 0 0 0 0 0 0 1 0 1 1 1 1 1 0 0 1 1 1 0 1 1 1 0 1 0 1 1 1 0 1 1 1 0 0 1 1 1 1
    1 0 1 0\\
0 0 0 0 0 0 0 0 0 1 0 1 1 0 0 1 1 1 0 1 0 0 1 0 0 1 1 0 1 0 1 0 0 0 1 0 0 1 0 0 0 1 0 1 0 1 1 0 0 1 0 0 1 0 1 1 1 0 0 1
    1 0 1 0
\end{array}\right)
$$
\item $[128,14]_{10}$ matrix $A$:
$$
{\tiny\left(\begin{array}{cccccccccccccccccccccccccccccccc}
1 0 0 0 0 0 0 0 0 0 0 0 0 0 0 1 0 1 0 0 1 0 0 1 0 1 1 1 0 0 1 0 1 1 0 1 0 1 1 1 1 0 0 0 1 1 1 0 1 1 0 1 0 0 1 0 0 0 0 1
    0 0 0 1 0 1 1 1 1 1 0 0 1 0 0 1 1 1 1 1 0 1 0 0 0 1 0 0 0 0 1 0 0 1 0 1 1 0 1 1 1 0 0 0 1 1 1 1 0 1 0 1 1 0 1 0 0 1
    1 1 0 1 0 0 1 0 0 1\\
0 1 0 0 0 0 0 0 0 0 0 0 0 0 1 0 0 0 1 1 0 1 1 0 0 1 0 1 1 1 0 0 1 1 0 0 0 1 0 0 1 1 0 1 1 0 1 0 1 1 0 0 1 1 0 1 0 0 1 0
    1 0 1 0 0 1 0 0 0 1 1 1 0 1 1 1 0 0 0 1 0 0 1 0 1 0 1 0 0 1 0 1 1 0 0 1 1 0 1 0 1 1 0 1 1 0 0 1 0 0 0 1 1 0 0 1 1 1
    0 1 0 0 1 1 0 1 1 0\\
0 0 1 0 0 0 0 0 0 0 0 0 0 0 0 1 0 0 0 1 1 0 1 1 0 0 1 0 1 1 1 0 0 1 1 0 0 0 1 0 0 1 1 0 1 1 0 1 0 1 1 0 0 1 1 0 1 0 0 1
    0 1 0 1 0 0 1 0 0 0 1 1 1 0 1 1 1 0 0 0 1 0 0 1 0 1 0 1 0 0 1 0 1 1 0 0 1 1 0 1 0 1 1 0 1 1 0 0 1 0 0 0 1 1 0 0 1 1
    1 0 1 0 0 1 1 0 1 1\\
0 0 0 1 0 0 0 0 0 0 0 0 0 0 1 0 0 0 0 1 1 1 1 1 0 1 1 1 0 0 1 0 1 0 0 1 1 1 1 0 0 0 1 0 1 0 1 1 0 0 0 1 0 1 1 1 0 1 1 0
    1 0 0 0 0 1 1 0 1 0 0 0 1 1 1 0 0 0 1 0 1 1 0 0 0 0 1 0 1 1 0 1 1 1 0 1 0 0 0 1 1 0 1 0 1 0 0 0 1 1 1 1 0 0 1 0 1 0
    0 1 1 1 0 1 1 1 1 1\\
1 0 0 0 0 0 0 0 0 0 1 1 1 0 0 1 1 1 0 1 0 1 0 1 1 1 0 0 1 1 1 0 0 0 0 0 0 0 0 0 1 0 0 0 0 0 1 0 1 1 1 1 1 0 0 1 0 1 1 0
    1 1 0 0 1 1 0 1 0 1 0 0 1 1 1 1 1 1 1 0 1 0 0 1 0 0 1 0 0 1 0 1 1 1 1 1 1 1 0 0 1 0 1 0 1 1 0 0 1 1 0 1 1 0 1 0 0 1
    1 1 1 1 0 1 0 0 0 0\\
0 1 0 0 0 0 0 0 0 0 0 0 1 1 0 1 0 1 0 1 1 0 0 0 0 0 0 0 0 0 0 1 0 0 1 0 0 1 1 0 1 0 0 1 0 1 1 1 0 0 0 1 0 1 0 1 1 1 1 0
    0 1 0 0 0 1 0 1 0 0 0 1 1 1 1 1 0 1 1 0 0 0 0 1 1 0 1 1 1 1 1 0 0 0 1 0 1 0 0 0 1 0 0 1 1 1 1 0 1 0 1 0 0 0 1 1 1 0
    1 0 0 1 0 1 1 0 0 1\\
0 0 1 0 0 0 0 0 0 0 1 0 0 1 1 1 0 1 0 1 1 0 1 1 0 0 1 0 1 0 0 1 0 1 1 1 0 0 0 1 1 1 1 1 1 0 1 0 1 1 1 1 1 1 0 0 0 1 1 1
    0 1 0 0 1 0 1 0 0 1 1 0 1 1 0 1 0 1 1 1 0 0 1 0 0 0 0 0 0 0 1 0 0 0 0 0 1 1 1 0 1 1 0 0 0 0 0 0 0 1 1 1 1 0 0 0 0 0
    0 0 1 1 0 1 1 1 0 0\\
0 0 0 1 0 0 0 0 0 0 1 0 1 1 1 0 0 0 0 1 0 1 0 0 1 1 1 0 0 0 1 1 0 0 1 1 0 0 0 1 1 1 0 0 1 0 1 0 0 0 0 1 1 1 0 1 0 0 0 0
    0 0 1 0 0 0 0 0 1 0 0 1 0 1 1 1 0 1 1 1 0 1 0 1 0 1 1 0 1 1 1 1 1 1 0 0 1 1 1 0 0 1 1 1 1 1 1 0 1 1 0 1 0 1 0 1 1 1
    0 1 1 1 0 1 0 0 1 0\\
0 0 0 0 1 0 0 0 0 0 1 0 0 1 0 1 1 1 0 1 1 1 0 1 0 1 0 1 1 0 1 1 1 1 1 1 0 0 1 1 1 0 0 1 1 1 1 1 1 0 1 1 0 1 0 1 0 1 1 1
    0 1 1 1 0 1 0 0 1 0 0 0 0 0 1 0 0 0 0 0 0 1 0 1 1 1 0 0 0 0 1 0 1 0 0 1 1 1 0 0 0 1 1 0 0 1 1 0 0 0 1 1 1 0 0 1 0 1
    0 0 0 0 1 1 1 0 1 0\\
0 0 0 0 0 1 0 0 0 0 0 1 0 1 1 1 1 1 0 0 1 0 1 1 0 1 1 0 0 1 1 0 1 0 1 0 0 1 1 1 1 1 1 1 0 1 0 0 1 0 0 1 0 0 1 0 1 1 1 1
    1 1 1 0 0 1 0 1 0 1 1 0 0 1 1 0 1 1 0 1 0 0 1 1 1 1 1 0 1 0 0 0 0 0 1 0 0 0 0 0 0 0 0 0 1 1 1 0 0 1 1 1 0 1 0 1 0 1
    1 1 0 0 1 1 1 0 0 0\\
0 0 0 0 0 0 1 0 0 0 1 1 0 0 1 1 0 1 1 1 1 0 0 1 0 0 0 0 1 1 1 1 1 0 0 0 0 1 0 0 1 1 1 1 0 1 1 0 0 1 1 0 0 0 1 0 0 0 0 0
    0 0 1 0 1 0 0 0 0 0 1 1 0 1 0 1 0 1 1 1 0 1 0 0 1 0 0 1 0 0 1 1 1 0 1 1 0 1 1 1 0 0 1 0 0 1 0 0 1 0 1 1 1 0 1 0 1 0
    1 1 0 0 0 0 0 1 0 1\\
0 0 0 0 0 0 0 1 0 0 1 0 0 0 1 1 1 0 1 1 0 0 1 0 1 1 0 1 1 1 1 0 1 0 1 1 1 0 1 0 0 1 1 0 1 0 1 0 1 0 1 1 1 1 1 0 0 1 0 1
    1 1 1 0 0 0 0 1 0 0 0 0 1 1 1 1 0 1 0 0 1 1 1 1 1 0 1 0 1 0 1 0 1 1 0 0 1 0 1 1 1 0 1 0 1 1 1 1 0 1 1 0 1 0 0 1 1 0
    1 1 1 0 0 0 1 0 0 1\\
0 0 0 0 0 0 0 0 1 0 0 1 0 0 1 1 0 1 0 0 1 0 1 1 1 0 0 0 1 0 1 0 1 1 1 1 0 0 1 0 0 0 1 0 1 0 0 0 1 1 1 1 1 0 1 1 0 0 0 0
    1 1 0 1 1 1 1 1 0 0 0 1 0 1 0 0 0 1 0 0 1 1 1 1 0 1 0 1 0 0 0 1 1 1 0 1 0 0 1 0 1 1 0 0 1 0 0 1 0 0 0 0 0 0 0 0 0 0
    1 1 0 1 0 1 0 1 1 0\\
0 0 0 0 0 0 0 0 0 1 1 1 0 0 0 1 0 0 1 1 1 0 0 1 0 1 1 1 1 0 0 1 1 0 1 0 1 1 1 0 0 0 0 1 1 0 0 0 0 1 0 1 0 1 1 0 1 1 1 1
    1 0 1 1 0 1 0 1 0 0 0 0 1 1 0 0 0 0 1 1 1 0 1 0 1 1 0 0 1 1 1 1 0 1 0 0 1 1 1 0 0 1 0 0 0 1 1 1 0 0 0 0 0 0 0 0 0 0
    0 1 0 0 1 1 0 0 1 0
\end{array}\right)}
$$
\end{itemize}
\end{appendix}



\end{document}